\documentclass[letterpaper, 10pt, conference,onecolumn]{ieeeconf}  

\IEEEoverridecommandlockouts                              

\usepackage{amsmath} 
\usepackage{amssymb}  
\usepackage{graphicx} 
\usepackage{mathptmx} 
\usepackage{times} 
\usepackage{datetime}
\newtheorem{lem}{Lemma}
\newtheorem{thm}{Theorem}

\def\<{\leqslant}           
\def\>{\geqslant}           

\def\d{\partial}

\def\Re{\mathrm{Re}}   

\def\col{\mathrm{vec}}   
\def\cH{\mathcal{H}}   
\def\mA{\mathbb{A}}    
\def\mR{\mathbb{R}}    

\def\Tr{\mathrm{Tr}}       
\def\rT{\mathrm{T}}        
\def\sym{\mathrm{sym}}        
\def\asym{\mathrm{asym}}        

\def\bE{\mathbf{E}}    

\def\[[[{[\![\![}
\def\]]]{]\!]\!]}

\def\bra{{\langle}}
\def\ket{{\rangle}}

\def\Bra{\big\langle}
\def\Ket{\big\rangle}

\def\rd{{\rm d}}        

\def\bJ{\mathbf{J}}

\def\x{\times}
\def\ox{\otimes}
\def\op{\oplus}

\def\cZ{{\mathcal Z}}

\def\cF{\mathcal{F}}
\def\cW{\mathcal{W}}

\def\cX{\mathcal{X}}

\def\cK{\mathcal{K}}

\def\cD{\mathcal{D}}

\def\cC{\mathcal{C}}

\def\cA{\mathcal{A}}
\def\cB{{\cal B}}
\def\cE{{\mathcal E}}

\def\mU{\mathbb{U}}

\def\mH{\mathbb{H}}
\def\mS{\mathbb{S}}

\def\eps{\epsilon}

\def\ups{\upsilon}

\def\argmin{\mathop{\mathrm{arg\, min}}}    
\def\Argmin{\mathop{\mathrm{Arg\, min}}}    
\title{\bf A Gradient Descent Approach to Optimal Coherent Quantum LQG Controller Design}

\author{Arash Kh. Sichani, \qquad Igor G. Vladimirov, \qquad Ian R. Petersen
\thanks{This work is supported by the Australian Research Council. The authors are with UNSW Canberra, ACT 2600, Australia. E-mail: {\tt arash\_kho@hotmail.com, igor.g.vladimirov@gmail.com, i.r.petersen@gmail.com}.}
}

\begin{document}

\maketitle

\thispagestyle{empty}

\begin{abstract}
    This paper is concerned with the Coherent Quantum Linear Quadratic Gaussian (CQLQG) control problem of finding a stabilizing measurement-free quantum controller for a quantum plant so as to minimize an infinite-horizon mean square performance index for the fully quantum closed-loop system. In comparison with the observation-actuation structure of classical controllers, the coherent quantum feedback is less invasive to the quantum dynamics and quantum information. Both the plant and the controller are open quantum systems whose dynamic variables satisfy the canonical commutation relations (CCRs) of a quantum harmonic oscillator and are governed by linear quantum stochastic differential equations (QSDEs). In order to correspond to such oscillators, these QSDEs must satisfy physical realizability (PR) conditions, which are organised as quadratic constraints on the controller matrices and reflect the preservation of CCRs in time. The CQLQG problem is a constrained optimization problem for the steady-state quantum covariance matrix of the plant-controller system satisfying an algebraic Lyapunov equation. We propose a gradient descent algorithm equipped with adaptive stepsize selection for the numerical solution of the problem. The algorithm finds a local minimum of the LQG cost over the parameters of the Hamiltonian and coupling operators of a stabilizing PR quantum controller, thus taking the PR constraints into account. A convergence analysis of the proposed algorithm is presented. A numerical example of a locally optimal CQLQG controller design is provided to demonstrate the algorithm performance.
\end{abstract}

\section{INTRODUCTION}

    Coherent quantum feedback control \cite{L_2000,M_2008} is a relatively novel quantum control paradigm which is aimed at achieving given performance specifications for quantum systems, such as internal stability and optimization of a cost functional. Such systems arise naturally in quantum physics \cite{H_2001} and its engineering applications (for example, nanotechnology and quantum optics \cite{GZ_2004}). The dynamic variables of quantum systems are (usually noncommuting) operators on an underlying Hilbert space which evolve according to the laws of quantum mechanics \cite{M_1998}. The latter make the quantum dynamics particularly sensitive to interaction with classical devices in the course of quantum measurement, as reflected in the projection postulate of quantum mechanics. In order to overcome this issue, coherent quantum control employs the idea of direct interconnection of quantum systems to be controlled (quantum plants) with other quantum systems playing the role of controllers, possibly mediated by light fields.  Unlike the traditional observation-actuation control loop, this fully quantum measurement-free feedback avoids the loss of quantum information as a result of its conversion to classical signals.

    Quantum-optical components, such as optical cavities, beam splitters and phase shifters, make it possible to implement coherent quantum feedback governed by linear quantum stochastic differential equations (QSDEs) \cite{P_1992,P_2010}, provided the latter are physically realizable (PR) as open quantum harmonic oscillators \cite{EB_2005,GZ_2004}. The resulting PR conditions \cite{JNP_2008,SP_2009,SP_2012} are organized as quadratic constraints on the coefficients of the QSDEs. The PR constraints for the state-space matrices of a coherent quantum controller complicate the solution of quantum counterparts to the classical $\cH_{\infty}$ and Linear Quadratic Gaussian (LQG) control problems.

    The Coherent Quantum LQG (CQLQG) control problem \cite{NJP_2009} seeks for a stabilizing PR quantum controller so as to minimize a mean square performance criterion for the fully quantum closed-loop system. A numerical procedure for finding \emph{suboptimal} controllers for this problem was proposed in \cite{NJP_2009}, and algebraic equations for the \textit{optimal} CQLQG controller were obtained in \cite{VP_2013a}. Despite the previous results, the CQLQG control problem does not lend itself to an ``elegant'' solution (for example, in the form of decoupled Riccati equations as in the classical case \cite{KS_1972}) and remains a subject of research. Since the main difficulties are caused by the coupling of the equations due to the PR constraints, a conversion of the CQLQG control problem to an \emph{unconstrained} problem by using Lagrange multipliers was considered in \cite{VP_2013b} for a related coherent quantum filtering problem which is a simplified feedback-free version of the CQLQG control problem.

    In the present paper, we develop an algorithm for the numerical solution of the CQLQG control problem by using the gradient descent method and the Hamiltonian parameterization of PR quantum controllers \cite{VP_2013a}. The latter is a different technique to handle the PR constraints  by reformulating the CQLQG control problem in an unconstrained fashion.
    More precisely, the optimal solution is sought in the class of stabilizing PR controllers whose state-space matrices are parameterized in terms of the free Hamiltonian and coupling operators of an open quantum harmonic oscillator \cite{EB_2005}. We obtain ordinary differential equations (ODEs) for the gradient descent in the Hilbert space of matrix-valued parameters. For this purpose, Fr\'{e}chet differentiation  is used together with related algebraic techniques \cite{BH_1998,M_1988,SIG_1998,VP_2010b,VP_2013a} to employ the analytic structure of the LQG cost as a composite function of the matrix-valued variables, involving Lyapunov equations.
    The advantage of the proposed
    approach is that, at intermediate steps, it produces stabilizing PR quantum controllers which can be used as gradually improving suboptimal solutions of the CQLQG control problem, and a locally optimal solution (if it exists) is achieved asymptotically by moving along negative gradient directions with a suitable choice of stepsizes. To this end, we provide an algorithm for adaptive selection of the stepsize for each iteration based on the second-order G\^{a}teaux derivative of the LQG cost along the gradient. However, the proposed gradient descent algorithm for the CQLQG control problem requires  for its initialization a \emph{stabilizing} PR quantum controller. Finding such a controller for an arbitrary  given quantum plant is a nontrivial open problem. Because of the lack of a systematic  solution for this \emph{quantum stabilization problem} at the moment,  the current version of the algorithm is initialized at a stabilizing PR quantum controller obtained by random search in the space defined by the Hamiltonian parameterization of PR controllers. Although a random search for an admissible starting point is acceptable for low-dimensional problems,  the development of a more systematic solution for this issue is a subject of future research.

    The paper is organised as follows.  Section~\ref{sec:not} outlines the notation used in the paper. Sections~\ref{sec:plant} and \ref{sec:controller} specify the quantum plants and coherent quantum controllers being considered. Section \ref{sec:PR} revisits PR conditions for linear quantum systems. Section~\ref{sec:problem} formulates the CQLQG control problem. Section~\ref{sec:Grad_Dec} describes a gradient descent system for finding local minima in the control problem. Section~\ref{sec:Grad_Algo} describes an algorithmic implementation of the gradient descent method.
    Section~\ref{subsec:con_analys} discusses convergence of the algorithm.
    Section~\ref{sec:example} provides a numerical example of designing a locally optimal CQLQG controller. Section~\ref{sec:Conclusion} gives concluding remarks. Appendices~\ref{app:FreDeriv} and \ref{app:sec_Gateaux} provide a subsidiary material on the differentiation of the LQG cost.

\section{NOTATION}\label{sec:not}

Vectors are assumed to be organized as columns unless specified otherwise, and the transpose $(\cdot)^{\rT}$ acts on matrices with operator-valued entries as if the latter were scalars. For a vector $X$ of operators $X_1, \ldots, X_r$ and a vector $Y$ of operators $Y_1, \ldots, Y_s$, the commutator matrix
$
    [X,Y^{\rT}]
    :=
    XY^{\rT} - (YX^{\rT})^{\rT}
$
is an $(r\x s)$-matrix whose $(j,k)$th entry is the commutator
$
    [X_j,Y_k]
    :=
    X_jY_k - Y_kX_j
$ of the operators $X_j$ and $Y_k$.
Furthermore, $(\cdot)^{\dagger}:= ((\cdot)^{\#})^{\rT}$ denotes the transpose of the entry-wise operator adjoint $(\cdot)^{\#}$. When it is applied to complex matrices,  $(\cdot)^{\dagger}$ reduces to the complex conjugate transpose  $(\cdot)^*:= (\overline{(\cdot)})^{\rT}$. Denoted by $\sym(\cdot):= \frac{(\cdot)+(\cdot)^\rT}{2}$ and $\asym(\cdot):= \frac{(\cdot)-(\cdot)^\rT}{2}$ are the symmetrizer and antisymmetrizer of matrices.  Also, we denote by $\mS_r$, $\mA_r$
 and
$
    \mH_r
    :=
    \mS_r + i \mA_r
$
the subspaces of real symmetric, real antisymmetric and complex Hermitian  matrices of order $r$, respectively, with $i:= \sqrt{-1}$ the imaginary unit. Furthermore, $I_r$ denotes the identity matrix of order $r$, positive (semi-) definiteness of matrices is denoted by ($\succcurlyeq$) $\succ$, and $\ox$ is the tensor product of spaces or operators (in particular, the Kronecker product of matrices).
The adjoints and  self-adjointness of linear operators acting on matrices is understood in the sense of the Frobenius inner product
$
    \bra M,N\ket
    :=
    \Tr(M^*N)
$ of real or complex matrices, with the corresponding Frobenius norm $\|M\|:= \sqrt{\bra M,M\ket}$ which is the standard Euclidean norm $|\cdot|$ for vectors.
Also, $\bE X := \Tr(\rho X)$ denotes the quantum expectation of a quantum variable $X$ (or a vector of such variables) over a density operator $\rho$ which specifies the underlying quantum state.

\section{QUANTUM PLANT}\label{sec:plant}

The quantum plant under consideration is an open quantum stochastic system which is coupled to another such system (playing  the role of a controller), with the dynamics of both systems being affected by the environment.
Both the plant and the controller are assumed to satisfy the physical realizability (PR) conditions \cite{JNP_2008,NJP_2009,SP_2009} which will be described in Section~\ref{sec:PR}.
The  plant has $n$ dynamic variables $x_1(t), \ldots, x_n(t)$ (with $n$ even) which are self-adjoint operators on a Hilbert space specified below. With the time arguments being omitted for brevity, the evolution of the plant state vector $x:= (x_k)_{1\< k \< n}$ and its contribution to a $p_1$-dimensional output of the plant  $y:= (y_k)_{1\< k\< p_1}$ (also with self-adjoint operator-valued entries) are governed by  QSDEs
\begin{equation}
\label{x_y}
    \rd x
    =
    A x\rd t  +  B \rd w + E \rd \eta,
    \qquad
    \rd y
    =
    C x\rd t  +  D \rd w.
\end{equation}
Here,
$
    A\in \mR^{n\x n}$,
$
    B\in \mR^{n\x m_1}$,
$
    C\in \mR^{p_1\x n}$,
$
    D\in \mR^{p_1\x m_1}$,
$
    E\in \mR^{n\x p_2}
$
are given constant matrices. Also,
\begin{equation}
\label{z}
    z
    :=
     Cx
\end{equation}
is a ``signal part'' of the plant output $y$, and $\eta$ is a $p_2$-dimensional output of the controller to be described in Section~\ref{sec:controller}. The external noise acting on the plant is represented by a quantum Wiener process $w:= (w_k)_{1\< k\< m_1}$ whose entries are self-adjoint operators on  a boson Fock space $\cF_1$ \cite{P_1992} with the quantum It\^{o} table
$        \rd w \rd w^{\rT}
        =
    \Omega_1
    \rd t$,
    where the matrix $\Omega_1 \in \mH_{m_1}$ is given by
$
    \Omega_1
    :=
    I_{m_1}+iJ_1
    \succcurlyeq 0
$.
Here,  the matrix $J_1 \in \mA_{m_1}$ specifies the CCRs between the entries of the plant noise $w$ as
$    [\rd w, \rd w^{\rT}]
    =
    2iJ_1\rd t
$
and (assuming that  the dimension $m_1$ is even) is given by
$    J_1 := I_{m_1/2} \ox \bJ$,
    with
    $ \bJ
    :=
    {\begin{bmatrix}
        0 & 1\\
        -1 & 0
    \end{bmatrix}}
$.

\section{QUANTUM CONTROLLER}\label{sec:controller}

Consider an interconnection of the plant (\ref{x_y})  with a coherent (that is, measurement-free) quantum controller. The latter  is also an open quantum system with an $n$-dimensional state vector $\xi:= (\xi_k)_{1\< k\< n}$ of self-adjoint operators on another Hilbert space, which also evolve in time. The assumption that the controller has the same number of dynamic variables as the plant is adopted from the classical LQG control theory. The controller interacts with the plant and the environment according to the QSDEs
\begin{equation}
\label{xi_eta}
    \rd \xi
     =
    a\xi\rd t + b \rd \omega + e\rd y,
    \qquad
    \rd \eta
     =
    c\xi \rd t + d\rd \omega.
\end{equation}
Here,
$
    a \in \mR^{n\x n}
$,
$
    b\in \mR^{n\x m_2}
$,
$
    c\in \mR^{p_2\x n}
$,
$
    d\in \mR^{p_2\x m_2}
$,
$
    e\in \mR^{n\x p_1}
$
are also constant matrices.
Similarly to (\ref{z}), the $p_2$-dimensional process
\begin{equation}
\label{zeta}
    \zeta
     :=
    c\xi
\end{equation}
is the signal part of the controller output $\eta$. The process $\omega$ in (\ref{xi_eta})  is a quantum noise which effects the controller and is an $m_2$-dimensional quantum Wiener process (with $m_2$ even) on another boson Fock space $\cF_2$ with the quantum Ito table
$    \rd \omega \rd \omega^{\rT}
    =
    \Omega_2\rd t$,
    where the matrix $\Omega_2 \in \mH_{m_2}$ is given by
$
    \Omega_2
    :=
    I_{m_2}+iJ_2
    \succcurlyeq 0$.
Here, the matrix $J_2 \in \mA_{m_2}$ specifies the CCRs between the entries  of the controller noise $\omega$ as
$    [\rd \omega, \rd \omega^{\rT}]
    =
    2iJ_2\rd t
$
and is given by $
    J_2 := I_{m_2/2} \ox \bJ
$.
The plant and controller noises $w$ and $\omega$ act on different boson Fock spaces $\cF_1$ and $\cF_2$, respectively, and hence, commute with each other. Therefore, the combined quantum Wiener process
\begin{equation}
\label{cW}
    \cW
    :=
    {\begin{bmatrix}
        w\\
        \omega
    \end{bmatrix}}
\end{equation}
of dimension $m:= m_1+m_2$ acts on the tensor product space $\cF_1\ox \cF_2$ and has a block diagonal CCR matrix $J$:
\begin{equation}
\label{WCCR}
    [\rd \cW, \rd \cW^{\rT}]
    =
    2iJ\rd t,
    \qquad
    J
    :=
    {\begin{bmatrix}
        J_1 & 0\\
        0 & J_2
    \end{bmatrix}}.
\end{equation}
Furthermore, the external boson fields are assumed to be in the product vacuum state $\ups := \ups_1 \ox \ups_2$, and hence, are uncorrelated. The resulting quantum Ito table of the combined Wiener process $\cW$ in (\ref{cW}) is
\begin{equation}
\label{WW}
    \rd \cW \rd \cW^{\rT}
    =
    \Omega
    \rd t,
    \qquad
    \Omega
    :=
    I_m + iJ
    =
    \Omega^*\succcurlyeq 0.
\end{equation}
In the controller dynamics (\ref{xi_eta}), the matrix $b$ is the noise gain matrix, while $e$ plays the role of the observation gain matrix, although $y$ is not an observation signal in the classical control theoretic sense. Accordingly, the process $\zeta $ in (\ref{zeta}) corresponds to the classical actuator signal.
Now, the combined set of QSDEs
 (\ref{x_y}), (\ref{xi_eta}) describes the fully quantum closed-loop system shown in Fig.~\ref{fig:system}.
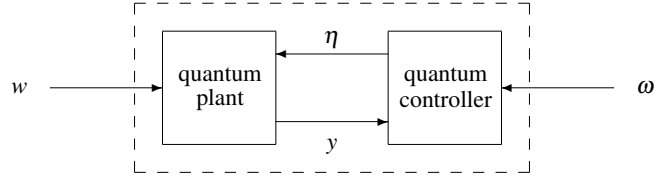
\begin{figure}[htpb]
\begin{center}
\unitlength=1.5mm
\linethickness{0.2pt}
\begin{picture}(50.00,26.00)
    \put(7.5,7.5){\dashbox(35,15)[cc]{}}
    \put(10.5,10){\makebox(10,12)[cc]{{\small quantum} }}
    \put(10.5,8){\makebox(10,12)[cc]{{\small plant} }}
    \put(10,10){\framebox(10,10)[cc]{ }}

    \put(30.5,10){\makebox(10,12)[cc]{{\small quantum} }}
    \put(30.5,8){\makebox(10,12)[cc]{{\small controller} }}
    \put(30,10){\framebox(10,10)[cc]{}}
    \put(0,15){\vector(1,0){10}}
    \put(50,15){\vector(-1,0){10}}
    \put(30,18){\vector(-1,0){10}}
    \put(20,12){\vector(1,0){10}}
    \put(-2,15){\makebox(0,0)[rc]{{\small$w$}}}
    \put(25,18.5){\makebox(0,0)[cb]{{\small$\eta$}}}
    \put(52,15){\makebox(0,0)[lc]{{\small$\omega$}}}
    \put(25,10.5){\makebox(0,0)[ct]{{\small$y$}}}
\end{picture}\vskip-5mm
\caption{The interconnected quantum plant and quantum controller form a fully quantum closed-loop system which is governed by  (\ref{x_y}), (\ref{xi_eta}) and is influenced by the environment through the quantum Wiener processes $w$, $\omega$.}
\label{fig:system}
\end{center}
\end{figure}
By using a  quadratic cost adopted in quantum control settings \cite{NJP_2009,VP_2013a} from classical LQG control \cite{KS_1972}, the performance of the coherent quantum controller will be described in Section~\ref{sec:problem} in terms of an $r$-dimensional quantum process
\begin{equation}
\label{cZ}
    \cZ
     :=
    F x   + G \zeta.
\end{equation}
Here, $F\in \mR^{r\x n}$, $G\in \mR^{r\x p_2}$ are given \emph{weighting matrices} whose entries quantify the relative importance of the state variables $x_1, \ldots, x_n$ of the plant and the ``actuator output'' variables $\zeta_1, \ldots, \zeta_{p_2}$ of the controller. The choice of $F$, $G$ is specified only by control design preferences and is not subjected to physical constraints. The process $\cZ$ in (\ref{cZ}) is linearly related to the $2n$-dimensional state vector
\begin{equation}
\label{cX}
    \cX :=
    {\begin{bmatrix}
        x\\
        \xi
    \end{bmatrix}}
\end{equation}
of the
closed-loop system whose dynamics are governed by the QSDE
\begin{equation}
\label{closed}
    \rd \cX
      =
      \cA\cX \rd t +   \cB\rd \cW ,
    \qquad
    \cZ
     =
      \cC\cX
\end{equation}
which is driven by the combined quantum Wiener process $\cW$ from (\ref{cW}).
The state-space matrices  $\cA \in \mR^{2n\x 2n}$, $\cB\in \mR^{2n\x m}$, $\cC\in \mR^{r\x 2n}$ of the closed-loop system in (\ref{closed}) are obtained by combining (\ref{x_y}), (\ref{xi_eta}) with (\ref{cW}), (\ref{cZ}), (\ref{cX}) and depend on the controller matrices $a$, $b$, $c$, $e$ in an affine fashion:
\begin{equation}
\label{cABC}
    \cA:=
    {\begin{bmatrix}
                  A       & E c\\
                  e C     & a
        \end{bmatrix}},
        \qquad
    \cB
    :=
    {\begin{bmatrix}
        B       & Ed\\
        e D     & b
    \end{bmatrix}},
    \qquad
    \cC
    :=
    {\begin{bmatrix}
                  F & G c
        \end{bmatrix}}.
\end{equation}

\section{CONDITIONS FOR PHYSICAL REALIZABILITY}\label{sec:PR}

Both the quantum plant (\ref{x_y}) and the coherent quantum controller (\ref{xi_eta}) are assumed to be physically realizable as open quantum harmonic oscillators, with initial complex separable Hilbert spaces $\cH_1$, $\cH_2$. In particular, their dynamic variables (which are self-adjoint operators on the product space  $\cH_1\ox \cH_2\ox \cF_1\ox \cF_2$ at any subsequent moment of time $t>0$) satisfy CCRs
\begin{equation}
\label{xxx}
    [x,x^{\rT}] = 2i\Theta_1,
    \qquad
    [\xi,\xi^{\rT}] = 2i\Theta_2,
    \qquad
    [x,\xi^{\rT}] = 0,
\end{equation}
where $\Theta_1, \Theta_2 \in \mA_n$ are constant nonsingular matrices. An equivalent form of  the CCRs for the combined vector $\cX$ from (\ref{cX}) is
\begin{equation}
\label{Theta}
    [\cX,\cX^{\rT}]
    =
    2i\Theta,
    \qquad
    \Theta
    :=
    {\begin{bmatrix}
        \Theta_1 & 0 \\
        0 & \Theta_2
    \end{bmatrix}}.
\end{equation}
The preservation of the CCRs (\ref{xxx}) (including the commutativity between $x$ and $\xi$) is a consequence of the unitary evolution of the isolated system formed from the plant, controller and their environment. The QSDE in (\ref{closed}) preserves the  CCR matrix $\Theta$ in (\ref{Theta}) in time if and only if the matrices $\cA$, $\cB$ in (\ref{cABC}) satisfy
\begin{equation}
\label{ThetaLyap}
    \cA \Theta + \Theta \cA^{\rT}
    +
    \cB J \cB^{\rT} = 0,
\end{equation}
where $J$ is the CCR matrix of the combined quantum Wiener process $\cW$ in (\ref{cW}) given by (\ref{WCCR}).  The relation (\ref{ThetaLyap}) is obtained by taking the imaginary part of the algebraic Lyapunov equation (ALE)
\begin{equation}
\label{SLyap}
    \cA S + S \cA^{\rT}
    +
    \cB \Omega \cB^{\rT} = 0
\end{equation}
(provided $\cA$ is Hurwitz)
for the steady-state quantum covariance matrix
\begin{equation}
\label{S}
    S
    :=
    \lim_{t\to +\infty }
    \bE (\cX(t)\cX(t)^{\rT})
    =
    P+i\Theta
    =
    S^*\succcurlyeq 0,
\end{equation}
with $\Omega$ the quantum Ito matrix from (\ref{WW}).
Here, the quantum expectation $\bE(\cdot)$ is taken over the product state $\varpi\ox \ups$, where $\varpi$ is the initial quantum state of the plant and controller on $\cH_1\ox \cH_2$, and $\ups$ is the vacuum state of the external fields on $\cF_1\ox \cF_2$. We have also used the convergence $\lim_{t\to +\infty}\bE \cX(t) = 0$ which is ensured by $\cA$ being Hurwitz. The real part
\begin{equation}
\label{P}
    P:= \Re S
\end{equation}
of the quantum covariance matrix $S$ from (\ref{S}) is the unique solution to the ALE
\begin{equation}
\label{PLyap}
    \cA P + P \cA^{\rT} + \cB \cB^{\rT} = 0,
\end{equation}
obtained by taking the real part of (\ref{SLyap}),
and coincides with the controllability Gramian \cite{KS_1972} of the pair $(\cA, \cB)$. Since the left-hand side of (\ref{ThetaLyap}) is an antisymmetric matrix of order $2n$, then,
by computing the diagonal $(n\x n)$-blocks  and the upper off-diagonal block of this matrix with the aid of (\ref{cABC}), it follows that the preservation of the CCR matrix $\Theta$ in (\ref{Theta}) is equivalent to
\begin{align}
\label{CCR11}
    A \Theta_1 + \Theta_1 A^{\rT} + B J_1 B^{\rT} + E d J_2 d^{\rT}E ^{\rT} &= 0,\\
\label{CCR22}
    a  \Theta_2 + \Theta_2 a ^{\rT} + e  D  J_1 D ^{\rT}e ^{\rT} + b  J_2 b ^{\rT}  &= 0,\\
\label{CCR12}
    (\Theta_1 C ^{\rT} + B  J_1 D ^{\rT})e ^{\rT}
    +
    E (c  \Theta_2 + d  J_2 b ^{\rT})
    & =  0;
\end{align}
cf. \cite[Eqs. (18)--(20)]{VP_2011b}.
Therefore, the fulfillment  of the equalities
\begin{align}
\label{CCR12_plant}
    \Theta_1 C^{\rT}      + BJ_1 D^{\rT}
    & =  0,\\
\label{CCR12_cont}
    c  \Theta_2 + d  J_2 b ^{\rT}& =  0
\end{align}
is sufficient for (\ref{CCR12}). Note that (\ref{CCR11}), (\ref{CCR12_plant}) are the conditions for physical realizability (PR) \cite{JNP_2008,NJP_2009,SP_2009} of the quantum plant which describe the preservation of the CCR matrix $\Theta_1$ in (\ref{xxx}) and $[x,y^{\rT}] = 0$. Similarly, the relations (\ref{CCR22}), (\ref{CCR12_cont}), which describe  the preservation of the CCR matrix $\Theta_2$ in (\ref{xxx}) and $[\xi, \eta^{\rT}] = 0$, are the PR conditions for the coherent quantum controller.
%
The PR condition
(\ref{CCR22}) can be regarded as a linear equation with respect to the matrix $a$, and its general solution is representable as
\begin{equation}
\label{a}
    a
    =
    2\Theta_2 R
    -\frac{1}{2}(e  D  J_1 D ^{\rT}e ^{\rT} + b  J_2 b ^{\rT})\Theta_2^{-1}.
\end{equation}
Here, the matrix $R \in \mS_n$ specifies the free Hamiltonian $\frac{1}{2}\xi ^{\rT} R  \xi$ which the PR controller would have in the absence of interaction with its surroundings; cf. \cite[Eqs. (20)--(22) on pp. 8--9]{EB_2005}.
The other PR condition  (\ref{CCR12_cont}) allows the matrix $c $ to be expressed in terms of $b $ as
\begin{equation}
\label{c}
    c
    =
    -d
    J_2
    b ^{\rT}
    \Theta_2^{-1}.
\end{equation}
The coupling between the output matrix $c$  and the noise gain matrix $b$ makes the design of a coherent quantum controller (\ref{xi_eta}) substantially different from that of the classical controllers even at the level of achieving internal stability for the closed-loop system.  Indeed, if the additional quantum noise $\omega$  is effectively eliminated from the state dynamics of the quantum controller by letting $b=0$, then (\ref{c}) implies that $c=0$, and hence, the matrix $\cA$ in (\ref{cABC}) becomes block lower triangular. In this case, the closed-loop system in (\ref{closed}) cannot  be internally stable if $A$ is not Hurwitz.
Also note that, in the formulations of the PR conditions \cite{JNP_2008,NJP_2009,SP_2012} for the plant and controller QSDEs (\ref{x_y}), (\ref{xi_eta}), the noise feedthrough matrices are usually given by $D = \begin{bmatrix}I_{p_1} & 0\end{bmatrix}$ and $d = \begin{bmatrix}I_{p_2} & 0\end{bmatrix}$, with $
    p_1\< m_1$ and
$
    p_2\< m_2
$. Such matrices $D$ and $d$ have full row rank and satisfy
\begin{equation}
\label{DDdd}
    DD^{\rT} = I_{p_1},
    \qquad
    dd^{\rT} = I_{p_2}.
\end{equation}
The full row rank property of $D$ corresponds to nondegeneracy of measurements in the classical setting, where $y$ in (\ref{x_y}) is an observation process. Furthermore, since $\det J_2\ne 0$ and $\det\Theta_2 \ne 0$, the full row rank property of $d$ implies that the map $\mR^{n\x m_2}\ni b\mapsto c\in \mR^{p_2\x n}$, given by (\ref{c}), is onto. This allows the matrix $c$ to be assigned any value by an appropriate choice of $b$, which plays a part in the stabilization issue mentioned above. Although (\ref{DDdd}) simplifies the algebraic manipulations, it is the rank properties of the matrices $D$, $d$  that are most important.

\section{COHERENT QUANTUM LQG CONTROL PROBLEM}\label{sec:problem}

    Following \cite{NJP_2009,VP_2013a}, we formulate the CQLQG control problem as that of minimizing the steady-state mean square value
    \begin{equation}
    \label{cE}
        \cE
        :=
        \frac{1}{2}
        \lim_{t\to+\infty}
            \bE
            (
                \cZ(t) ^{\rT} \cZ(t)
            )
            =
            \frac{1}{2}
        \Bra
            \cC ^{\rT}\cC ,
            P
        \Ket
        \longrightarrow
        \min
    \end{equation}
    for the criterion process $\cZ$ of the closed-loop system  (\ref{closed}) over internally stabilizing (that is, making the matrix $\cA$ Hurwitz) PR quantum controllers (\ref{xi_eta}) of fixed dimensions described in Sections \ref{sec:controller}, \ref{sec:PR}.
    Here,
    $
        \cZ^{\rT}\cZ
        =
        \sum_{k=1}^r
        \cZ_k^2
    $
    is the sum of squared entries of $\cZ$ (and hence, $\cZ^{\rT}\cZ$ is a positive semi-definite self-adjoint operator) and $P$ is the controllability Gramian of the closed-loop system given by (\ref{P}), (\ref{PLyap}).
    The LQG cost $\cE$ in (\ref{cE}) is a function of the triple
    \begin{equation}
    \label{u}
        u:= (R,b,e)
        \in
        \mS_n
        \x
        \mR^{n\x m_2}
        \x
        \mR^{n\x p_1}
        =: \mU
    \end{equation}
    which parameterizes PR quantum controllers  (\ref{xi_eta}) through (\ref{a}), (\ref{c}), with the controller noise feedthrough matrix
    $
        d\in \mR^{p_2\x m_2}
    $
    being fixed and satisfying (\ref{DDdd}).  Accordingly, the minimization in (\ref{cE}) is carried out over the set
    \begin{equation}
    \label{U0}
        \mU_0:= \{u \in \mU:\ \cA\ {\rm in}\ (\ref{cABC})\ {\rm is\ Hurwitz}\}
    \end{equation}
    of those $u$ which specify internally stabilizing PR quantum controllers   for the quantum plant (\ref{x_y}).
    For what follows, the set $\mU$ on the right-hand side of (\ref{u})
    is endowed with the structure of a Hilbert space with the direct sum inner product
    $    \bra (R,b,e),(R',b',e') \ket := \bra R,R'\ket + \bra b,b'\ket + \bra e,e'\ket$. Note that $\mU_0$ in (\ref{U0}) is an open subset of $\mU$.
\section{GRADIENT FLOW FOR THE LQG COST}\label{sec:Grad_Dec}

     The gradient descent approach to the solution of the CQLQG control problem is to move with the negative gradient flow for the LQG cost function $\cE$ in (\ref{cE}) towards a local minimum. The gradient descent can be regarded as a dynamical system governed by the ODE
     \begin{equation}
        \label{gradient_sys}
        \dot{u}(\tau) = - g(u(\tau)), \qquad u(0) = u_0.
     \end{equation}
     Here, $\dot{(\, )}:= \d_{\tau}(\cdot)$ is the derivative with respect to fictitious time $\tau\> 0$, and the gradient
     \begin{equation}
        \label{grad}
        g(u): = \d_u \cE(u) = (\d_R \cE, \d_b \cE, \d_e \cE)
     \end{equation}
     is the Fr\'{e}chet derivative of the LQG cost with respect to $u$ in the Hilbert space $\mU$ associated with the Hamiltonian parameterization of PR quantum controllers in (\ref{u}). More precisely, the map $g: \mU_0 \to \mU$ is well-defined on the open set $\mU_0$ in (\ref{U0}). The starting point in (\ref{gradient_sys}) is assumed to satisfy
     \begin{equation}
     \label{u0}
       u_0 := (R_0, b_0, e_0) \in \mU_0,
     \end{equation}
     so that the corresponding PR controller is internally stabilizing.
         Unless $u_0$ is a stationary point of $\cE$,       the LQG cost is strictly decreasing along the trajectory of the ODE (\ref{gradient_sys}) in view of
     $
        \cE(u(\tau))^{^\bullet} = -\|g(u(\tau))\|^2
     $.
     The first-order Fr\'{e}chet derivative in (\ref{grad}) is computed in the following lemma. For its formulation, we denote by $Q$ the observability Gramian of the pair $(\cA,\cC)$ which is a unique solution of the ALE
          \begin{equation}
            \label{QLyap}
                \cA^{\rT} Q + Q \cA + \cC^{\rT} \cC =0,
            \end{equation}
            provided the matrix $\cA$ in (\ref{cABC}) is Hurwitz. Furthermore, we will use the Hankelian of the closed-loop system defined by
            \begin{equation}
            \label{H}
                H:= QP.
            \end{equation}
            Also, we partition $(2n\x 2n)$-matrices $X$ (such as $P$, $Q$, $H$) into $(n\x n)$-blocks  $X_{jk}$ as
$
    X
                :=
                {\small\begin{array}{cc}
                {}_{\leftarrow n \rightarrow}  {}_{\leftarrow n\rightarrow} &\\
                        {\small\begin{bmatrix}
                        X_{11} & X_{12}\\
                        X_{21} & X_{22}
           \end{bmatrix}}
            &\!\!\!\!\!
                {\small\begin{matrix}
                    \updownarrow\!{}^n\\
                    \updownarrow\!{}_n
            \end{matrix}}\\
                {}
            \end{array}}
$.

     \begin{lem}
          \label{lem:1st_Frechet}
        For any $u \in \mU_0$ from (\ref{U0}), the Fr\'{e}chet derivative (\ref{grad}) of the LQG cost $\cE$ in (\ref{cE}) can be computed as
          \begin{align}
                \label{eq:frech_der_R}
                \d_R \cE &= -2\sym(\Theta_2 H_{22}),\\
                \label{eq:frech_der_b}
                \d_b \cE &= Q_{21}Ed+Q_{22}b-\psi b J_2 - \chi d J_2, \\
                \label{eq:frech_der_e}
                \d_e \cE &= H_{21}C^{\rT} + Q_{21} BD^{\rT} + Q_{22} e - \psi e D J_1 D^{\rT}.
          \end{align}
          Here, $\psi$ and $\chi$ are auxiliary $(n\x n)$-matrices defined by
          \begin{equation}
          \label{psi_chi}
            \psi := \asym (H_{22} \Theta_2^{-1}),\qquad
            \chi := \Theta_2^{-1}(H_{12}^{\rT} E + P_{21}F^{\rT} G  + P_{22} c^{\rT} G^{\rT}G),
          \end{equation}
          with $P$, $Q$, $H$ the Gramians and Hankelian from (\ref{PLyap}), (\ref{QLyap}), (\ref{H}). \hfill $\square$
     \end{lem}

     The proof of Lemma~\ref{lem:1st_Frechet} is similar to that of \cite[Theorem~1]{VP_2013a} and is given in Appendix~\ref{app:FreDeriv} for completeness.
     That the trajectories of the gradient descent system in (\ref{gradient_sys}) will not ``miss'' local minima of the LQG cost is justified by the following lemma.
     \begin{lem}
        \label{lem:stability}
        A point $u_*  \in \mU_0$ in (\ref{U0}) is a stable equilibrium of the ODE (\ref{gradient_sys}) if and only if it is a local minimum of the LQG cost $\cE$ in (\ref{cE}).\hfill$\square$
     \end{lem}
\begin{proof}
The assertion of the lemma can be established by using \cite[Theorem~3]{absil2006stable} and the \emph{analyticity} \cite{H_1990} (rather than infinite differentiability) of the LQG cost $\cE$ in a neighbourhood of any point $u\in \mU_0$. The analyticity follows from the representation
     \begin{equation}
    \label{rat}
        \cE 
         = -\frac{1}{2} \col(\cC^{\rT} \cC)^{\rT}(\cA \op \cA)^{-1}\col(\cB\cB^{\rT})
     \end{equation}
which is obtained from (\ref{PLyap}), (\ref{cE}) by using the column-wise vectorization $\col(\cdot)$   of matrices \cite{M_1988,SIG_1998}  and the Kronecker sum $\cA \op \cA:= I_{2n}\ox \cA + \cA \ox I_{2n}$ of the matrix $\cA$ with itself. Indeed, the representation (\ref{rat}) implies that $\cE$ is a rational function of the entries of $\cA$, $\cB$, $\cC$ in (\ref{cABC}) which, in turn, are polynomial functions of the entries of $R$, $b$, $e$ in view of (\ref{a}), (\ref{c}), and hence, $\cE(u)$ is a rational function of $u$. Therefore, the function $\cE(u)$ is analytic on the open set $\mU_0$  since the matrix $\cA\op \cA$ is also Hurwitz (and hence, nonsingular) for any $u \in \mU_0$.
\end{proof}

     In practice, the gradient descent ODE (\ref{gradient_sys}) is solved by using a numerical algorithm, which is the subject of the next section.

\section{GRADIENT DESCENT ALGORITHM}\label{sec:Grad_Algo}

We will now consider a numerical algorithm which implements the gradient descent method (\ref{gradient_sys}) for the CQLQG control problem in the form \begin{equation}
    \label{unext}
            u_{k+1}:=u_k - s_k g(u_k),
            \qquad
            k = 0,1,2,\ldots.
\end{equation}
This recurrence equation is initialized with matrices $R_0$, $b_0$, $e_0$ of an internally stabilizing PR controller in (\ref{u0}) (see Section~\ref{sec:Initialization}). The gradient $g(u_k)$ is computed by using Lemma~\ref{lem:1st_Frechet}, and   the stepsize $s_k>0$ is chosen as described in Section~\ref{sec:Stepsize}. The iterations in (\ref{unext}) are stopped when a termination condition is satisfied (see Section~\ref{subsec:Term_Cond}). The ingredients of the algorithm are discussed in the subsequent sections.
%
%
%

\subsection{Initialization}\label{sec:Initialization}

    The gradient descent algorithm (\ref{unext}) relies on existence of an internally stabilizing PR quantum controller which can be used as an initial point. As mentioned in Introduction, the existence of such controllers (that is, nonemptiness of the set $\mU_0$ in (\ref{U0})) for a given quantum plant  (and a systematic method of finding them) remains an open problem. In the present version of the algorithm, this quantum stabilization problem is solved by using a random search in the finite-dimensional Hilbert space $\mU$ in (\ref{u}).

\subsection{Stepsize selection}\label{sec:Stepsize}

According to the conventional limited minimization rule  (see, for example,  \cite{bertsekas99nonlinear}), the stepsize $s_k$ is chosen for each iteration of the gradient descent by solving the minimization problem
    \begin{equation}
        \label{equ:step_min}
        s_k\in \Argmin_{0\< s\< h_k} \cE(u_k-s g(u_k))
    \end{equation}
with a \emph{constant} search horizon $h_k:= h>0$. Here, we use the convention that $\cE(u):=+\infty$ if $u\ne \mU_0$ (thus discarding those controllers which are not internally stabilizing).
A restricted version of the line search with a constant horizon $h$ may suffer from the inability to adapt properly to the behaviour of the function $\cE$ in its minimization over the ray $\{u_k-sg(u_k):\ s \>0\} \subset \mU$. In order to overcome this issue, for the stepsize selection in the gradient descent algorithm (\ref{unext}), we will use  a modified version of the limited minimization rule  with an adaptive choice of the search horizon $h_k$ in each iteration.
More precisely, $h_k$ can be chosen so as to enable (\ref{equ:step_min}) to ``capture'' the minimum of $\cE$ over the whole ray if $\cE$ is a strictly convex quadratic function. To this end, consider the quadratically truncated Taylor series
    \begin{equation}
    \label{quad}
        \cE(u-sg) = \cE(u) - s\cD_g \cE + \frac{s^2}{2} \cD_g^2 \cE + o(s^2),
    \end{equation}
where
        \begin{align}
        \label{Gat1}
            \cD_v \cE(u) & := \d_s\cE(u+s v)\big|_{s=0} = \bra g(u), v\ket,\\
        \label{Gat2}
            \cD_v^2 \cE(u) & :=\d_s^2\cE(u+s v)\big|_{s=0} = \bra \d_u^2 \cE(u)(v), v\ket
        \end{align}
are the first and second-order G\^{a}teaux (or directional) derivatives \cite{LS_1961} of the LQG cost at a point $u \in \mU_0$ (specifying an internally  stabilizing controller) along $v \in \mU$. Here, in view of (\ref{Gat1}),
        \begin{equation}
        \label{gg}
            \cD_g\cE = \|g\|^2 = \|\d_R\cE\|^2 + \|\d_b\cE\|^2 + \|\d_e\cE\|^2\>0.
        \end{equation}
Also, $\d_u^2 \cE(u):=\d_u g(u)$ in (\ref{Gat2}) is the second-order Fr\'{e}chet derivative of $\cE$ which is a self-adjoint operator on the Hilbert space $\mU$ in (\ref{u}) whose computation is outlined in Appendix~\ref{app:sec_Gateaux}.
Now, if $\cD_g^2 \cE(u)>0$, then the quadratic polynomial of $s$ on the right-hand side of (\ref{quad}) (with the higher-order terms being neglected) achieves its unique minimum at a nonnegative value of $s$:
\begin{equation}
\label{smin}
  \argmin_{s \> 0} \left(\frac{s^2}{2} \cD_g^2\cE- s\cD_g \cE \right) = \frac{\cD_g\cE}{\cD_g^2\cE} = \frac{\|g\|^2}{\cD_g^2\cE}.
\end{equation}
This suggests using the right-hand side of (\ref{smin}) as a search horizon $h_k$ in (\ref{equ:step_min}), provided $\cD_g^2 \cE(u)>0$. However, if the latter inequality does not hold, the argument, based on a quadratic approximation of the minimization problem (\ref{equ:step_min}), is no longer valid and needs to be amended. In this case (when  $\cD_g^2 \cE(u)\< 0$), the search horizon can be chosen so as to avoid the domination of nonlinear terms over the linear term in the quadratically truncated Taylor series for the LQG cost along the ideal gradient descent trajectory in (\ref{gradient_sys}):
\begin{align}
\nonumber
    \cE(u(\tau+s)) &= \cE(u(\tau)) +(\cE(u))^{^\bullet}s + (\cE(u))^{^{\bullet\bullet}}\frac{s^2}{2}  + o(s^2)\\
\label{quad2}
    & =
    \cE(u(\tau)) - \|g\|^2s +  \cD_g^2\cE s^2 + o(s^2),
\end{align}
where (\ref{Gat2}) is used.
For $s\> 0$, the comparison of the absolute values $\|g\|^2s$ and $|\cD_g^2\cE| s^2$ of the linear and quadratic terms in (\ref{quad2}) shows that the latter does not dominate the former if
\begin{equation}
\label{abs}
    s\< \frac{\|g\|^2}{|\cD_g^2\cE|}.
\end{equation}
This inequality is closely related to the accuracy of (\ref{unext}) as Euler scheme for numerical integration of the ODE (\ref{gradient_sys}). More precisely, if the stepsizes $s_k>0$ in (\ref{unext}) are significantly smaller than the respective values of the right-hand side of (\ref{abs}), then $u_k$ becomes an accurate approximation of the ideal gradient descent trajectory $u(\tau)$ at fictitious  time $\tau:= s_0 + \ldots +s_{k-1}$. A combination of (\ref{smin}) and (\ref{abs}) justifies the following heuristic rule for choosing the search horizon at the current point $u_k \in \mU_0$:
\begin{equation}
\label{hor}
    h_k:= \min\left(h_{\max},\, \frac{\|g\|^2}{|\cD_{g}^2\cE|}\Big|_{u=u_k}\right).
\end{equation}
Here, $h_{\max}$ is a given positive threshold which becomes active, for example, if $\cD_{g}^2\cE$ vanishes.
The stepsize selection algorithm considered below, replaces the minimization problem in (\ref{equ:step_min}) with a different procedure which involves a finite subset of values of $s$ from a geometric progression
\begin{equation}
\label{skl}
  s_{k,\ell}:= h_k f^{\ell},
  \qquad
  \ell = 0,1,2,\ldots
\end{equation}
whose initial value $h_k$ is given by (\ref{hor}). The common ratio $f\in (0,1)$ is a parameter of the algorithm which affects how ``densely'' the progression fills the interval $[0,h_k]$. Now, the adaptive stepsize selection algorithm proceeds as follows:
\begin{equation}
\label{sk}
    s_k := s_{k,j},
\end{equation}
where the $j$th element of the geometric progression (\ref{skl}) is chosen according to the Armijo rule \cite{bertsekas99nonlinear} with a parameter $\sigma \in (0,1)$:
\begin{equation}
\label{Arm}
                j :=
                \min\big\{
                \ell\> 0:\,
                \cE(u_k)-\cE(u_k - s_{k,\ell} g(u_k)) \>  \sigma s_{k,\ell} \|g(u_k)\|^2
                \big\}.
\end{equation}
Here, the subset of indices $\ell$ is nonempty since $\sigma <1$ and $\lim_{\ell \to +\infty}s_{k,\ell} = 0$.
The inequality in (\ref{Arm}) is important in the convergence analysis of the gradient descent algorithm. In particular, the condition $\sigma >0$ ensures that $\cE(u_k)$ is strictly decreasing until $u_k$ achieves a stationary point of the LQG cost. Such a point is a stable equilibrium of the gradient descent only if it delivers a local minimum to the LQG cost.
\subsection{Termination condition}\label{subsec:Term_Cond}

    Since the gradient descent sequence $u_k$ in (\ref{unext}) can converge to a local minimum of the LQG cost only asymptotically, as $k\to +\infty$, the algorithm is equipped with a termination condition (for stopping the iterations) which reflects the proximity to the limit point.  More precisely, we use the following termination condition which employs the relative smallness of the gradient as specified by a dimensionless parameter $\eps>0$:
    \begin{equation}
    \label{stop}
        s_k\|g(u_k)\| \< \eps \|u_k\|.
    \end{equation}

\section{CONVERGENCE OF THE GRADIENT DESCENT ALGORITHM}\label{subsec:con_analys}

    As the proposed algorithm is based on the classical gradient descent approach, its convergence analysis follows a similar reasoning, which we provide below for completeness.

    \begin{thm}
        \label{thm:converg}
        Suppose $(u_k)_{k\>0}$ is the gradient descent sequence  in (\ref{unext}) with the stepsize selection described by (\ref{hor})--(\ref{Arm}).
        Then every limit point $u_* \in \mU_0$ of this sequence is a stationary point of the LQG cost $\cE$, that is, $g(u_*) =0$.
    \end{thm}
        \begin{proof}
Since the sequence $\cE(u_k)\>0$ is nonincreasing, it has a finite limit. Therefore, (\ref{unext}), (\ref{sk}) and the Armijo rule in (\ref{Arm}) imply that
$
    \sigma s_k \|g(u_k)\|^2 \< \cE(u_k)-\cE(u_{k+1})\to 0$ as $k\to +\infty$. Hence, in view of $\sigma>0$,  it follows that
    \begin{equation}
    \label{sg0}
        \lim_{k\to +\infty} \big(s_k \|g(u_k)\|^2\big) = 0.
    \end{equation}
Now,  assume that the gradient descent sequence  $u_k$ has a limit point $u_* := \lim_{\cK\ni  k\to +\infty}u_k\in \mU_0$ such that $g(u_*)\ne 0$, where $\cK:= \{0\< k_1 < k_2 < \ldots\}$ is an infinite subset of nonnegative integers which specifies the respective subsequence of $u_k$. Then  the analyticity of the LQG cost on the open set  $\mU_0$ implies that
\begin{align}
\label{gpos}
    \lim_{\cK\ni  k\to +\infty}g(u_k) & = g(u_*)\ne 0,\\
\label{hpos}
    \lim_{\cK\ni  k\to +\infty} h_k
    & =
    \min\left(h_{\max},\, \frac{\|g(u_*)\|^2}{|\cD_{g}^2\cE(u_*)|}\right)>0,
\end{align}
where use is made of (\ref{hor}). Note that,
if $\cD_{g}^2\cE(u_*) = 0$, the limit in (\ref{hpos}) is equal to $h_{\max}>0$. A combination of (\ref{gpos}) with (\ref{sg0}) implies that
\begin{equation}
    \label{s0}
        \lim_{\cK\ni k\to +\infty} s_k = 0.
    \end{equation}
In turn, by combining (\ref{s0}) with (\ref{hpos}) and recalling (\ref{sk}) and the condition $0<f<1$, it follows that the indices
$
    j_p := \log_f \frac{s_{k_p}}{h_{k_p}} = \frac{\ln h_{k_p} - \ln s_{k_p}}{-\ln f}$
of the elements of the geometric progression in (\ref{skl}), which correspond to $k_p\in \cK$, diverge to infinity as
     $
    p\to +\infty
$,
and hence, $j_p\>1$ for all sufficiently large $p$. For all such $p$, the stepsize candidates $s_{k_p, j_p-1} = \frac{s_{k_p}}{f}$ do not pass the Armijo selection rule in (\ref{Arm}), that is,
            \begin{equation}
                \label{ineq:contra}
                \cE(u_k)-\cE\big(u_k-\frac{s_k}{f}g(u_k)\big) < \sigma \frac{s_k}{f}\|g(u_k)\|^2
            \end{equation}
for all sufficiently large $k \in \cK$.
Upon multiplying both parts of (\ref{ineq:contra}) by $\frac{f}{s_k}$ and taking the limit, this inequality leads to
\begin{align}
\nonumber
    \|g(u_*)\|^2
    & =
    \lim_{\cK \ni k\to +\infty}
    \Big(
        \frac{f}{s_k}
        \Big(
            \cE(u_k)-\cE\big(u_k-\frac{s_k}{f}g(u_k)\big)
        \Big)
    \Big)\\
\label{ggg}
    & \< \sigma \lim_{\cK\ni  k\to +\infty}\|g(u_k)\|^2 = \sigma \|g(u_*)\|^2,
\end{align}
where use is also made of (\ref{gpos}) and (\ref{s0}). However, since $\sigma <1$, the inequality in (\ref{ggg}) contradicts the assumption that $g(u_*)\ne 0$. This contradiction shows that any limit point $u_*\in \mU_0$ of the gradient descent sequence satisfies $g(u_*) = 0$.
%
        \end{proof}

    Note that the CQLQG control problem inherits a special type of symmetry from the LQG cost $\cE$ which is invariant under symplectic similarity transformations of the controller variables $\xi\mapsto \Sigma \xi$, with $\Sigma \in \mR^{n\x n}$ satisfying $\Sigma \Theta_2 \Sigma^{\rT} = \Theta_2$ and thus preserving the CCR matrix $\Theta_2$  (see, for example, \cite{VP_2011b,VP_2013a}). Hence,  the stationary points of the LQG cost are not isolated, which complicates the convergence rate analysis for the proposed gradient descent algorithm. This issue  is beyond the scope of the present study and will be addressed elsewhere by using more advanced analytic tools (such as in \cite{absil2009} and related references).

\section{A NUMERICAL EXAMPLE OF OPTIMAL CQLQG CONTROLLER DESIGN}\label{sec:example}

The gradient descent algorithm of Section~\ref{sec:Grad_Algo} was tested to find a locally optimal solution of the CQLQG control problem for a PR quantum plant in (\ref{x_y}) with dimensions $n=m_2=p_1=p_2=r=2$, $m_1 = 4$ and randomly generated state-space matrices $A$, $B$, $C$, $E$, satisfying the PR conditions (\ref{CCR11}), (\ref{CCR12_plant}), and the weighting matrices $F$, $G$ in (\ref{cZ}):
\begin{align*}
    A=&   {\begin{bmatrix}
                            0.9534 &  -1.1165\\
                           0.4193 &  1.8821
            \end{bmatrix}}\!,\ \ \qquad
    B=
    {\begin{bmatrix}
                      -1.7174  & -0.2189 &   1.9180&   0.5636\\
                       -0.6815 &  1.3570  &   0.2985  &  -0.3679
    \end{bmatrix}},\\
    C=&   {\begin{bmatrix}
                           -1.3570  &  -0.2189\\
                           -0.6815  &  1.7174
            \end{bmatrix}}\!,\qquad
    D=
             {\begin{bmatrix}
                    1 & 0 & 0 & 0\\
                    0 & 1 & 0 & 0
            \end{bmatrix}}\!,\qquad\qquad\quad
    E=      {\begin{bmatrix}
                         -0.3238 &  0.2779\\
                         -1.1693 &  -0.5966
            \end{bmatrix}},\\
    F= &      {\begin{bmatrix}
                        -0.8290 &  -0.9665\\
                        -1.8655&  -0.0357
            \end{bmatrix}}\!,\qquad
    G=      {\begin{bmatrix}
                        -0.2324 & -0.1608\\
                        -0.5822 & -1.0961
            \end{bmatrix}}\!,\qquad
    d=      {\begin{bmatrix}
                        1 & 0\\
                        0 & 1
            \end{bmatrix}}.
\end{align*}
This plant is unstable (the eigenvalues of the  matrix $A$ are $1.4177 \pm 0.5025 i$). The algorithm was run with parameters $h_{\max}=1$,
$f=0.5$, $\sigma = 0.9$, $\epsilon = {10}^{-6}$ in (\ref{hor})--(\ref{stop}) for $10$ randomly generated stabilizing PR controllers as initial points.
Starting from these points, it has taken $307$
to $2318$
steps for the algorithm to reach the fulfillment of the termination condition, with the average number of iterations being $1075$.
The local minimum value of the LQG cost is $\cE_{\min}=12.1026$
and is achieved at the following controller parameters:
\begin{align*}
    R=   {\begin{bmatrix}
                            -0.5611  &  -1.5567\\
                            -1.5567  & 1.8283
            \end{bmatrix}},
            \qquad
    b=   {\begin{bmatrix}
                            1.8111 &  0.7201\\
                           -1.4979 &  -3.9696
            \end{bmatrix}},
            \qquad
    e=      {\begin{bmatrix}
                            -0.1250  &  4.9673\\
                            -4.4929  &  -1.3387
            \end{bmatrix}}\!\!.
\end{align*}
The values of the LQG cost $\cE(u_k)$ for the gradient descent sequences $u_k$ are presented in Fig.~\ref{fig:NExp} in the form of semi-logarithmic graphs of $\frac{\cE(u_k)}{\cE_{\min}}-1$.
\begin{figure}[htbp]
    \centering
    \vskip-3mm    \includegraphics[width=9cm]{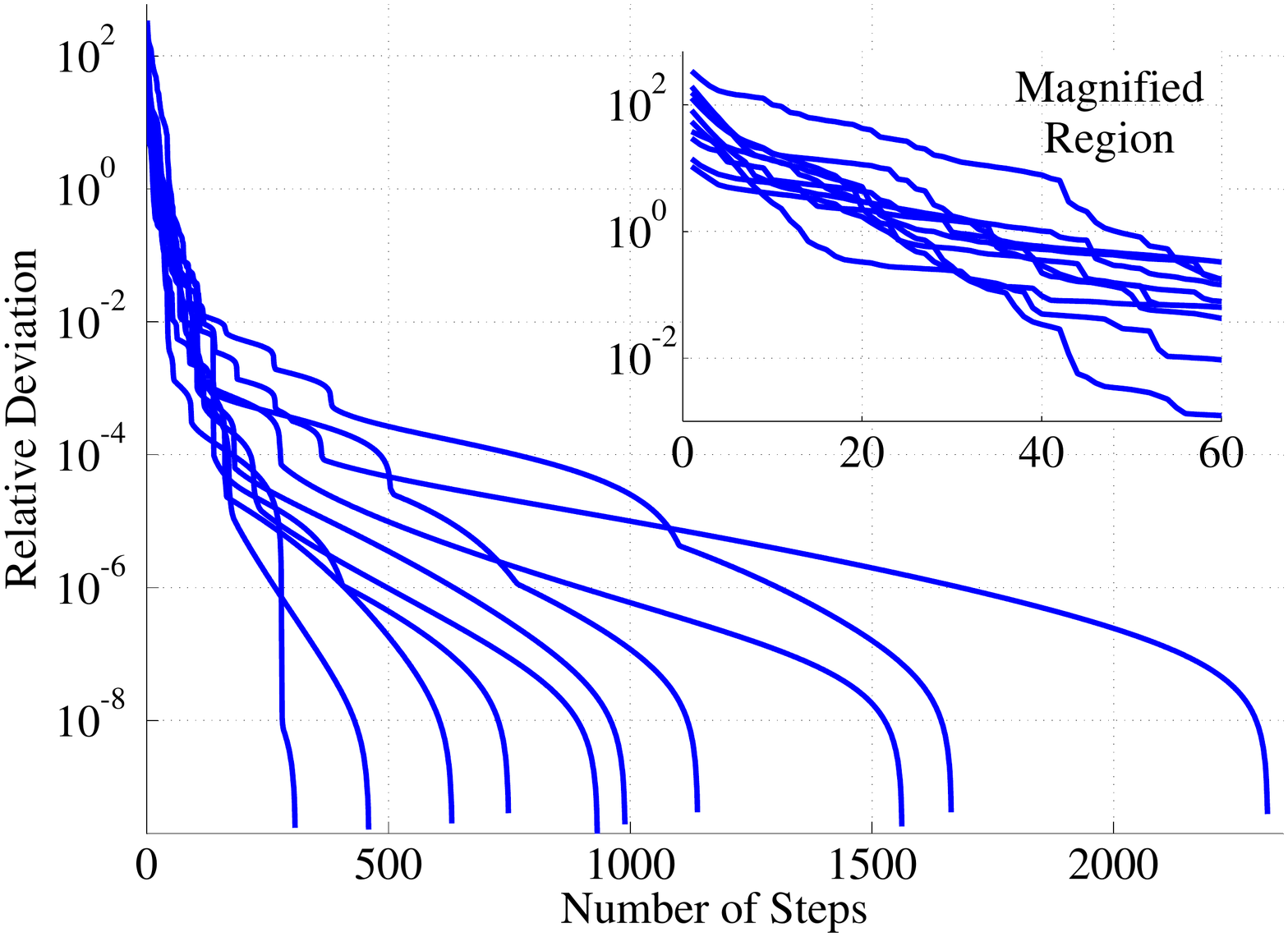}
    \vskip-3mm
    \caption{The relative deviations $\frac{\cE(u_k)}{\cE_{\min}}-1$ from the minimum value of the LQG cost  on a logarithmic scale versus the number of steps $k$.}\label{fig:NExp}
\end{figure}
These graphs are in qualitative agreement with the relatively slow linear convergence rate, typical for gradient descent methods. However, they also show that the proposed algorithm is fairly reliable, being able to cope with poor initial approximations where the LQG cost exceeds the minimum value by an order of magnitude.

\section{CONCLUSION}\label{sec:Conclusion}

A gradient descent algorithm has been developed for the numerical solution of the optimal CQLQG controller design problem, and its  convergence has been investigated. The algorithm has been tested and it appears to be fairly reliable in a numerical example with randomly generated stabilizing PR controllers  as initial points. The lack of a more systematic method for initialization and a relatively slow convergence rate are the main shortcomings of the algorithm. These issues are a subject for future research and will be tackled in subsequent publications.





%
%
%
%

\appendix
\renewcommand{\theequation}{\Alph{subsection}\arabic{equation}}

\subsection{Computing the Fr\'{e}chet derivative of the LQG cost}\label{app:FreDeriv}
\setcounter{equation}{0}

From the ALEs (\ref{PLyap}), (\ref{QLyap}) and the properties of the Frobenius inner product of matrices, it follows that the LQG cost $\cE$ in (\ref{cE}) is representable as
\begin{align}
    \label{eq:CQLQG_equi}
    \cE(u)
    =
    \frac{1}{2}\bra \cC^{\rT} \cC, P \ket
    =
    \frac{1}{2} \bra Q, \cB \cB^{\rT} \ket
    =
    - \bra H, \cA \ket,
\end{align}
where $H$ is the Hankelian given by (\ref{H}).
In what follows, $\delta(\cdot)$ denotes the first variation, and $\delta_X(\cdot)$ is the first variation with respect to an independent matrix-valued variable $X$. Since the matrix $R$ influences the LQG cost $\cE$ in (\ref{eq:CQLQG_equi}) only through the controller matrix $a$ in (\ref{a}) which enters the matrix $\cA$ in (\ref{cABC}),  then
\begin{align}
\nonumber
    \delta_R  \cE
     =&
    \frac{1}{2}
    \bra
        \cC^{\rT}\cC,
        \delta_R P
    \ket
    =
    -
    \frac{1}{2}
    \bra
        \cA^{\rT}Q+Q\cA,
        \delta_R P
    \ket\\
\nonumber
    =&
     -
    \frac{1}{2}
    \bra
        Q,
        \cA \delta_R P  + (\delta_R P)\cA^{\rT}
    \ket
     =
    \frac{1}{2}
    \bra
        Q,
        (\delta_R \cA) P  +  P\delta_R \cA^{\rT}
    \ket\\
\label{Aa}
    = &
    \bra
        H, \delta_R \cA
    \ket
    = \bra H_{22} , \delta_R a \ket
    =
    \bra H_{22} , 2 \Theta_2 \delta R \ket.
 \end{align}
 Now, since $R \in \mS_n$, and the subspaces $\mS_n$ and  $\mA_n$ are orthogonal, then the first variation of the LQG cost in (\ref{Aa}) takes the form
$
    \delta_R \cE
    = 2 \bra \Theta_2^{\rT} H_{22} , \delta R \ket
                    = -2 \bra \sym (\Theta_2 H_{22}) , \delta R \ket
$,
which, by the definition of the Fr\'{e}chet derivative, establishes (\ref{eq:frech_der_R}).
By a similar reasoning, the first variations of the LQG cost with respect to $b$ and $e$ are as follows:
\begin{align*}
    \delta_b \cE =&  \bra H , \delta_b \cA \ket + \bra Q \cB , \delta \cB \ket + \bra \cC P, \delta_b \cC \ket\\
                    =&  \bra H_{22} , \delta_b a \ket + \bra E^{\rT} H_{12} , \delta_b c \ket + \bra (Q \cB)_{22} , \delta b \ket \\
                     & + \bra G^{\rT} F P_{12} + G^{\rT} G c P_{22}, \delta_b c \ket \\
                    =& \bra -\asym(H_{22} \Theta_2^{-1})bJ_2 + Q_{21}Ed+ Q_{22}b \\
                     &  -  \Theta_2^{-1} ( H_{12}^{\rT} E + P_{21} F^{\rT} G + P_{22} c^{\rT} G^{\rT} G) d J_2 , \delta b \ket,\\
    \delta_e \cE =& \bra H_{21} , \delta e \ C \ket + \bra H_{22} , \delta_e a \ket + \frac{1}{2} \bra Q , \delta_e (\cB \cB^{\rT}) \ket \\
                    =& \bra H_{21} C^{\rT} , \delta e \ket + \bra H_{22}\Theta^{-1} , \asym(\delta e DJ_1 D^{\rT} e^{\rT}) \ket \\
                    &  +  \bra Q \cB , \delta_e \cB \ket \\
                    =& \bra H_{21} C^{\rT} - \asym(H_{22}\Theta^{-1})eDJ_1D^{\rT} \\
                    &  +(Q_{21}B+ Q_{22}eD)D^{\rT}, \delta e \ket,
\end{align*}
which leads to (\ref{eq:frech_der_b}) and (\ref{eq:frech_der_e}) in view of (\ref{DDdd}) and (\ref{psi_chi}).
\subsection{Computing the second-order G\^{a}teaux derivative of the LQG cost}\label{app:sec_Gateaux}
\setcounter{equation}{0}

The first-order G\^{a}teaux derivative of the LQG cost $\cE$ along the gradient $g$ in (\ref{grad}) is expressed in terms of the first-order Fr\'{e}chet derivatives from (\ref{eq:frech_der_R})--(\ref{eq:frech_der_e}) by using (\ref{gg}),
provided $u \in \mU_0$.
Hence, the second-order G\^{a}teaux derivative along the gradient can be computed as
\begin{align}
   \nonumber
   &\cD_g^2 \cE=  \cD_g (\|g\|^2) = 2\bra \cD_g g,g \ket\\
\label{eq:2nd_Gateaux}
   =& 2 (\bra \cD_g\d_R\cE,\d_R\cE \ket
                   +  \bra \cD_g\d_b\cE,\d_b\cE \ket
                +  \bra \cD_g\d_e\cE,\d_e\cE \ket) ,
\end{align}
where
\begin{align}
\label{cal:d2R}
    \cD_g\d_R\cE =&  -2 \sym(\Theta_2 \cD_g H_{22}),\\
\nonumber
     \cD_g\d_b\cE
     =&
     \cD_g Q_{21}Ed + \cD_g(Q_{22}b) \\
\label{cal:d2b}
     &  - \cD_g(\psi b) J_2 - \cD_g\chi d J_2,\\
\nonumber
    \cD_g\d_e\cE=&
    \cD_gH_{21}C^{\rT} + \cD_gQ_{21} BD^{\rT} \\
\label{cal:d2e}
    &  + \cD_g(Q_{22} e) - \cD_g(\psi e ) D J_1 D^{\rT}.
\end{align}
The second-order G\^{a}teaux derivative in (\ref{eq:2nd_Gateaux}) can now be computed by using (\ref{cal:d2R})--(\ref{cal:d2e}),
the Leibniz product rule
and the first-order G\^{a}teaux derivatives of $P$, $Q$.
By differentiating both sides of the ALE (\ref{PLyap})  and its dual (\ref{QLyap}), it follows that the matrices $\cD_g P$ and $\cD_g Q$ are unique solutions of the ALEs
\begin{align*}
    \cA \cD_g P + \cD_g P \cA^{\rT} + 2\sym(\cD_g \cA P + \cD_g \cB \cB^{\rT})& =0,      \\
    \cA^{\rT} \cD_g Q + \cD_g Q  \cA + 2\sym(\cA^{\rT} \cD_g Q + \cC^{\rT} \cD_g \cC)& =0.
\end{align*}
Here, in view of (\ref{cABC}), (\ref{a}), (\ref{c}) and  the relation $\cD_g u = g$,
\begin{align*}
    \cD_g \cA &= {\begin{bmatrix}
                    0  &  -E d J_2 \d_b \cE^{\rT} \Theta_2^{-1}\\
                    \d_e \cE C  & \cD_g a
                 \end{bmatrix}},\\ \nonumber
    \cD_g \cB &= {\begin{bmatrix}
                    0           & 0 \\
                    \d_e \cE D  & \d_b \cE
                 \end{bmatrix}},\\
    \cD_g \cC &= {\begin{bmatrix}
                     0 & -G d J_2 \d_b \cE^{\rT} \Theta_2^{-1}
                 \end{bmatrix}},
\end{align*}
with
$$
    \cD_g a = 2 \Theta_2 \d_R \cE - \asym (\d_e \cE D J_1 D^{\rT} e^{\rT} + \d_b \cE J_2 b^{\rT})\Theta_2^{-1}.
$$
\end{document}